%% file: arxiv.tex
\documentclass[12pt]{amsart}  

\usepackage[utf8]{inputenc}
\usepackage[T1]{fontenc}
\usepackage{graphicx}
\usepackage{longtable}
\usepackage{float}
\usepackage{latexsym}
\usepackage{amssymb}
\usepackage{bbm}
\usepackage{hyperref}
\usepackage{amsmath,amsthm}
\usepackage{dsfont}
\usepackage{color}
\usepackage{enumerate}

\newtheorem{theorem}{Theorem}[section]

\newtheorem{lemma}[theorem]{Lemma}
\newtheorem{definition}[theorem]{Definition}

\graphicspath{{figs/}}

\newcommand{\suchthat}{\mathbin |}
\newcommand{\absval}[1]{\left\vert #1 \right\vert}

\newcommand{\suc}{\mathrm{succ}}
\newcommand{\MSO}{\textrm{MSO}}
\newcommand{\FO}{\textrm{FO}}

\newcommand{\N}{\mathds{N}}
\newcommand{\cT}{\mathcal{T}}

\newcommand{\cV}{\mathcal{V}}
\newcommand{\topmin}{\preceq_{\mathrm{top}}}

\newcommand{\cN}{\mathcal{N}}

\let\phi=\varphi

\newcommand{\CCC}{\mathcal{C}}

 \newcommand{\TTT}{\mathcal{T}}
 \newcommand{\VVV}{\mathcal{V}}

\newcommand{\ssL}{\mathsf{L}}

\renewcommand{\MSO}{\textup{MSO}\xspace}

\usepackage{xspace}

\begin{document}

\title[succ-inv FO on Graphs with excl. top. subgraphs]{Successor-Invariant First-Order Logic on Graphs with Excluded
  Topological Subgraphs}
\author{Kord Eickmeyer}
\address{Technical University Darmstadt, Department of Mathematics, Schlossgartenstr. 7, 64289 Darmstadt, Germany}
\email{eickmeyer@mathematik.tu-darmstadt.de}

\author{Ken-ichi Kawarabayashi}
\address{National Inst. of Informatics and JST, ERATO, Kawarabayashi Large Graph Project, Hitotsubashi 2-1-2, Tokyo 101-8430, Japan}
\email{k\_keniti@nii.ac.jp}

\let\realbfseries=\bfseries
\def\bfseries{\realbfseries\boldmath}

\begin{abstract}
  We show that the model-checking problem for successor-invariant
  first-order logic is fixed-parameter tractable on graphs with
  excluded topological subgraphs when parameterised by both the size
  of the input formula and the size of the exluded topological
  subgraph. Furthermore, we show that model-checking for
  order-invariant first-order logic is tractable on coloured posets of
  bounded width, parameterised by both the size of the input formula
  and the width of the poset.

  Results of this form, i.e. showing that model-checking for a certain
  logic is tractable on a certain class of structures, are often
  referred to as algorithmic meta-theorems since they give a unified
  proof for the tractability of a whole range of problems. First-order
  logic is arguably one of the most important logics in this context
  since it is powerful enough to express many computational problems
  (e.g. the existence of cliques, dominating sets etc.) and yet its
  model-checking problem is tractable on rich classes of graphs. In
  fact, Grohe et al.~\cite{gks14} have shown that model-checking for
  $\FO$ is tractable on all nowhere dense classes of graphs.

  Successor-invariant $\FO$ is a semantic extension of $\FO$ by allowing
  the use of an additional binary relation which is interpreted as a
  directed Hamiltonian cycle, restricted to formulae whose truth value
  does not depend on the specific choice of a Hamiltonian cycle. While
  this is very natural in the context of model-checking (after all,
  storing a structure in computer memory usually brings with it a
  linear order on the structure), the question of how the
  computational complexity of the model-checking problem for this
  richer logic compares to that of plain $\FO$ is still open.

  Our result for successor-invariant $\FO$ extends previous results
  for this logic on planar graphs \cite{EngelmannKS12} and graphs with
  excluded minors \cite{ekk13}, further narrowing the gap between what
  is known for $\FO$ and what is known for successor-invariant
  $\FO$. The proof uses Grohe and Marx's structure theorem for graphs
  with excluded topological subgraphs \cite{gromar12+}. For
  order-invariant $\FO$ we show that Gajarský et al.'s recent result
  for $\FO$ carries over to order-invariant $\FO$.
\end{abstract}

\maketitle

\section{Introduction}
\label{sec:intro}

Model-checking is one of the core algorithmic problems in finite model
theory: Given a sentence $\varphi$ in some logic $\ssL$ and a finite
structure $A$, decide whether $A \models \varphi$. The problem can be
generalised by allowing $\varphi$ to have free variables, in which
case we would like to find instances $\bar a$ for which $A
\models \varphi[\bar a]$, or count the number of such
instances. One important application of this is the case where
$\varphi$ is a database query and $A$ the datebase to be queried. The
logic $\ssL$ from which $\varphi$ is drawn then serves as an abstract
model of the database query language.

Commonly studied logics $\ssL$ include first-order logic ($\FO$) and
monadic second-order logic ($\MSO$). Even for first-order logic the
model-checking problem is PSPACE complete already when restricted to
structures $A$ with two elements. On the other hand, for every fixed
$\FO$-formula $\varphi$, checking whether $A \models \varphi$ can be
done in time polynomial in the size of $A$. This discrepancy between
the \emph{query complexity}, i.e. the complexity depending on the size
of the query $\varphi$ on the one hand and the \emph{data complexity},
i.e. the complexity depending on the size of the structure $A$, on the
other hand suggests that the complexity of model-checking problems is
best studied in the framework of \emph{parameterised
  complexity}~\cite{DowneyF98,FlumG06}.

In parameterised complexity, apart from the size $n$ of the input
problem (commonly the length of an appropriate binary representation
of $\varphi$ and $A$) a \emph{parameter} $k$ is introduced. For
model-checking problems the size of the input formula is a common
choice of parameter. The role of PTIME as the class of problems
commonly considered to be tractable is played by the parameterised
complexity class of \emph{fixed-parameter tractable (fpt)} problems,
i.e. problems which can be solved in time
\[
f(k)\cdot n^c
\]
for some computable function $f$ and a constant $c$. Note that the
constant $c$ must not depend on $k$, and indeed the model-checking
problem for first-order logic is unlikely to be fixed-parameter
tractable.

In order to obtain tractable instances of model-checking problems, one
can restrict the space of admissible input structures $A$, e.g. by
requiring the Gaifman graph of $A$ to possess certain graph theoretic
properties such as bounded degree or planarity. A long list of results
have been obtained, starting with Courcelle's famous result that
model-checking for monadic second-order logic is fixed-parameter
tractable on structures $A$ with bounded
tree-width~\cite{Courcelle90}.

Results of this form are often referred to as \emph{algorithmic
  meta-theorems} because many classical problems can be rephrased as
model-checking problems by formalising them as a sentence $\varphi$ in
a suitable logic. For example, since the existence of a Hamiltonian
cycle in a graph $G$ of bounded tree-width can be expressed by a
sentence $\varphi$ of monadic second-order logic, Courcelle's Theorem
immediately implies that hamiltonicity can be checked in polynomial
time on such graphs. Besides giving a mere proof of tractability,
algorithmic meta-theorems provide a unified treatment of how
structural properties can be used in algorithm
design. Cf.~\cite{Grohe07} and \cite{Kreutzer11} for excellent surveys
of the field of algorithmic meta-theorems.

The model-checking problem for first-order logic is particularly well
studied and has been shown to be fixed-parameter tractable on a large
number of graph classes: Starting with Seese's result~\cite{Seese96}
for graphs of bounded degree, Frick and Grohe showed tractability on
classes of graphs with bounded tree-width and, more generally, locally
bounded tree-width~\cite{FrickGro01}, which in particular includes
planar graphs. This has been generalised to graph classes with
excluded minors~\cite{FlumG01} and locally excluded
minors~\cite{DawarGroKre07}. Using rather different techniques, {Dvo\v
  r\'ak} et al. gave a linear fpt model-checking algorithm for
first-order logic on graphs of bounded
expansion~\cite{DvorakKT13}. As a generalisation of all the
graph classes mentioned so far, Grohe et al. have shown
in~\cite{gks14} that model-checking for first-order logic is possible
in near-linear fpt on all nowhere dense graph classes.

While the tractability of model-checking for first-order logic on
sparse graphs is well understood, few results are available for
classes of dense graphs. Recently, Gajarský et al. gave an fpt
algorithm for $\FO{}$ model-checking on posets of bounded width, which
we extend to order-invariant $\FO$ in Section~\ref{sec:dense}.

\subsection*{Excluded Topological Subgraphs}

A more general concept than that of a class of graphs excluding some
graph $H$ as a minor is that of graphs which exclude $H$ as a
\emph{topological subgraph}. This is the concept originally used by
Kuratowski in his famous result that a graph is planar if, and only
if, it does not contain $K_5$ nor $K_{3,3}$ as a topological subgraph
(cf. Section 4.4 in~\cite{diestel}). Recently, Grohe and Marx have
extended Robertson and Seymour's graph structure theorem to classes of
graphs excluding a fixed graph $H$ as a topological
subgraph~\cite{gromar12+}: These graphs can be decomposed along small
separators into parts which exclude $H$ as a minor and parts in which
all but a bounded number of vertices have small degree.

Since every topological subgraph of a graph $G$ is also a minor of
$G$, if a class $\CCC$ of graphs excludes some graph $H$ as a
topological subgraph then it also excludes $H$ as a minor. The
converse is not true, however, since every $3$-regular graph excludes
$K_5$ as a topological subgraph, but for every $r \in \N$ there is a
$3$-regular graph containing $K_r$ as a minor. On the other hand,
graph classes with excluded topological subgraphs have bounded
expansion, so model-checking for first-order logic is tractable on
these classes by {Dvo\v r\'ak} et al.'s
result.

Figure~\ref{fig:sparseclasses} shows an overview of sparse graph
classes on which model-checking for first-order logic is
tractable. Note that a class $\CCC$ of graphs excludes some finite
graph $H$ as a topological subgraph if, and only if, there is an $r
\in \N$ such that $\CCC$ excludes the clique $K_r$ as a topological
subgraph.

\begin{figure}
  \begin{center}
    \resizebox{.8\textwidth}{!}{\includegraphics{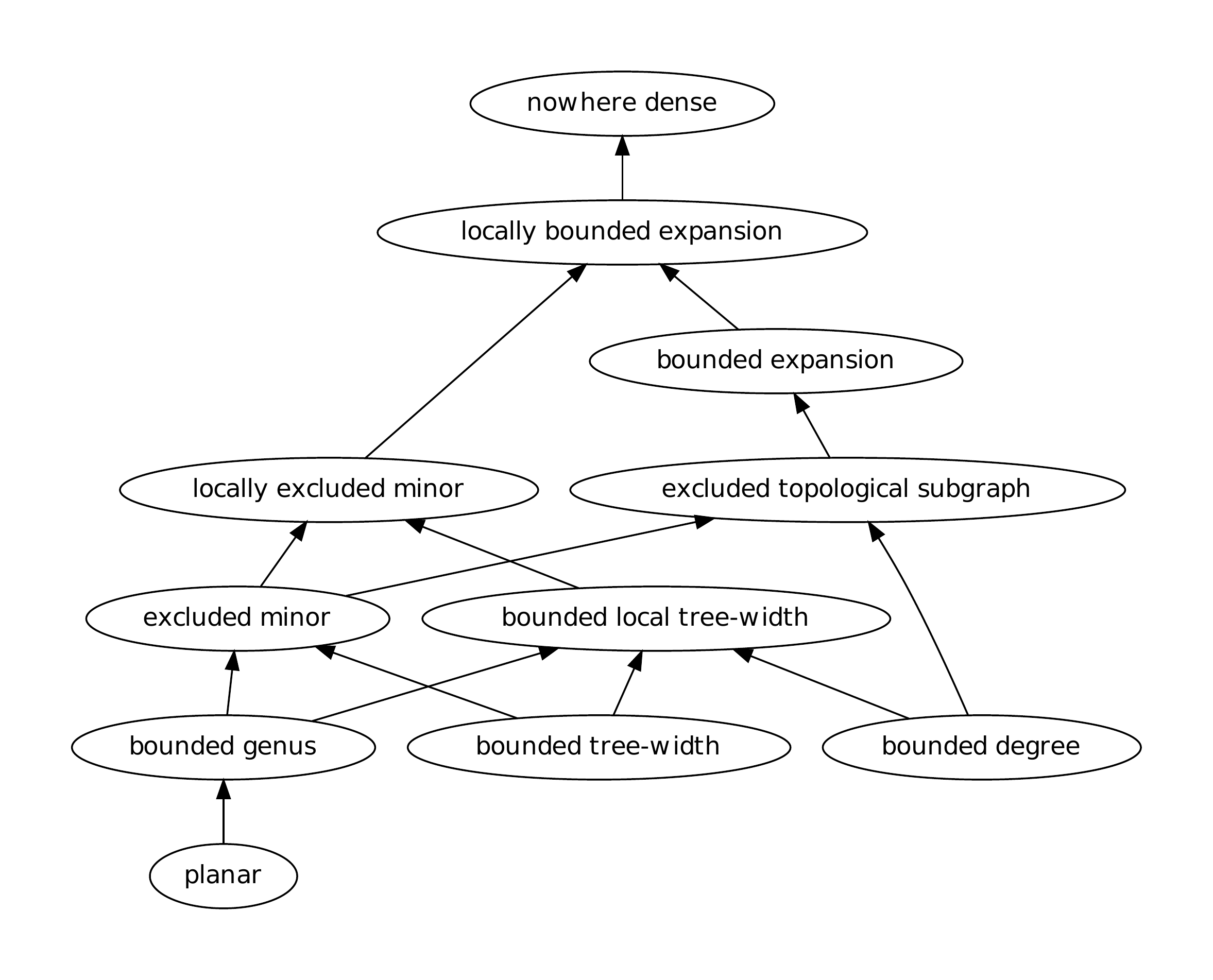}}
  \end{center}
  \caption{Sparse classes of graphs on which model-checking for
    first-order logic is tractable.}
  \label{fig:sparseclasses}
\end{figure}

\subsection*{Successor-Invariant Logic}

We investigate the question in how far tractability results for
first-order model-checking carry over to \emph{successor-invariant
  first-order logic}, i.e. first-order logic enriched by a binary
successor relation, restricted to formulae whose truth value does not
depend on the specific choice of successor relation. Linear
representations of an input structure $A$ to a model-checking
algorithm usually induce some linear order on the elements of $V(A)$,
and it seems natural to make this linear order (or at least its
successor relation) accessible to the query formula. This may,
however, break the structural properties of the Gaifman graph of $A$
needed by the model-checking algorithm.

Having access, even invariantly, to a successor relation provably
increases the expressive power of $\FO$ on finite structures, as shown
in~\cite{Rossman07}. However, all known classes of structures
separating $\FO$ from order-invariant or successor-invariant $\FO$
contain large cliques, and in fact on trees~\cite{BenediktS09} and on
structures of bounded tree-depth~\cite{EickmeyerEH14} even
order-invariant $\FO$ has the same expressive power as plain $\FO$. On
all the classes depicted in Figure~\ref{fig:sparseclasses}, this
question is still open, prompting for tractability results for
successor-invariant or even order-invariant $\FO$ on these classes.

Previous work investigating the complexity of model-checking for
successor-invariant first-order logic to that of plain first-order
logic has been carried out by \cite{EngelmannKS12}, who showed
tractibility on planar graphs, and \cite{ekk13}, who showed
tractability on graph classes with excluded minors. Here we extend
these results further by generalising from excluded minors to excluded
topological subgraphs, further narrowing the gap between what is known
for first-order logic and succesor-invariant first-order logic.

Note that for first-order logic, the result of \cite{gks14} is optimal
if one restricts attention to classes of graphs which are closed under
taking subgraphs. In fact, Kreutzer has shown in~\cite{Kreutzer11}
that under the complexity theoretic assumption that $\mathrm{FPT} \not
= \mathrm{W}[1]$, if model-checking for $\FO$ on some subgraph-closed
class $\CCC$ of graphs is fixed-parameter tractable, then $\CCC$ is
nowhere dense (see also~Section~1.4 of~\cite{DvorakKT13}). Examples of
classes of graphs on which model-checking is fpt even for monadic
second-order logic but which are not nowhere dense are graphs of
bounded clique-width~\cite{CourcelleMakRot00}.

\section{Preliminaries and Notation}
\label{sec:prelim}

For a natural number $n$ we let $[n]$ denote the interval
$\{1,\ldots,n\}$.

\subsection{Graphs}
\label{sec:graphprelim}

We will be dealing with finite simple (i.e. loop-free and without
multiple edges) graphs, cf.~\cite{diestel,tutte} for an in-depth
introduction. Thus a \emph{graph} $G = (V,E)$ consists of some finite
set $V$ of \emph{vertices} and a set $E \subseteq \binom{V}{2}$ of
\emph{edges}. We write $uv \in E$ for $\{u,v\} \in E$. For a set $U
\subseteq V$ we denote the \emph{induced subgraph} on $U$ by $G[U]$,
i.e. the graph $(U,E')$ with
\[
E' := \{ uv \suchthat u,v \in U\text{ and }uv \in E \}.
\]
For ease of notation we occasionally blur the distinction between a
set $U$ of vertices and the subgraph induced on this set. The union $G
\cup H$ of two graphs $G = (V,E)$ and $H = (U,F)$ is defined as the
graph $(U\cup V, E \cup F)$. For a set $U$ of vertices, $K[U]$ denotes
the complete graph (or \emph{clique}) with vertex set $U$. For $k \in
\N$, we denote the $k$-clique $K[[k]]$ by $K_k$.

A \emph{walk} is a sequence of vertices $v_1,\ldots,v_{\ell} \in V$,
alternatively written as a function $v:[\ell] \to V$, such that
$v_iv_{i+1} \in E$ for all $i=1,\ldots,\ell-1$. A \emph{path} is a
walk in which $v_i \not= v_j$ for $i \not= j$, except possibly $v_1 =
v_{\ell}$, in which case the path is called a \emph{cycle}. The vertices
$v_2,\ldots,v_{\ell-1}$ are called \emph{inner vertices}. Two paths
$v_1,\ldots,v_{\ell}$ and $w_1,\ldots,w_{m}$ are called
\emph{independent} if neither of them contains an inner vertex of the
other, i.e. if $v_i=w_j$ implies $i \in \{1,\ell\}$ and $j \in
\{1,m\}$.

For $k \geq 1$, a \emph{$k$-walk through} a graph $G = (V,E)$ is a
walk $w : [\ell] \to V$ such that 
\[
1 \leq \absval{ \{ i \in [\ell] \suchthat w(i) = v \} } \leq k
\]
for all $v \in V$.
A $1$-walk is also called a \emph{Hamiltonian path}.

\subsubsection*{Tree-Decompositions}

A \emph{tree} is a connected acyclic graph.
A \emph{tree-decomposition} of a graph $G = (V,E)$ is a pair
$(\cT,\cV)$ consisting of a tree $\cT = (T,F)$ and a mapping $\cV:T\to
2^V, t \mapsto \cV_t$ such that
\begin{itemize}
\item $\bigcup_{t \in T} \cV_t = V$,
\item for every edge $uv \in E$ there is a $t \in T$ with $u,v \in
  \cV_t$, and
\item for every $v \in V$ the set $\{ t \in T \suchthat v \in \cV_t \}$ is a
  subtree of $\cT$ (i.e. it is connected).
\end{itemize}
The sets $\cV_t$ are called the \emph{bags} of the
tree-decomposition. Let $t \in \cT$ have neighbours
$\cN(t) \subseteq \cT$. The \emph{torso} $\bar{\cV}_t$ of
$\cV_t$ is the graph
\[
G[\cV_t] \cup \bigcup_{u \in \cN(t)} K[\cV_t \cap \cV_u].
\]

Rather than decomposing a graph $G$ into small parts we are interested
in decompositions $(\TTT,\VVV)$ for which (the torsos of) all bags
$\VVV_t$ have nice structural properties and for which
\[
\absval{\VVV_s \cap \VVV_t}
\]
is small for all $s \not = t \in \TTT$. The \emph{(maximal) adhesion} of
$(\TTT,\VVV)$ is the maximum of $\absval{\VVV_s \cap \VVV_t}$ for all
$s \not= t \in \TTT$.

\subsubsection*{Subgraphs, Minors, Topological Subgraphs}

Let $G = (V,E)$ and $H = (W,F)$ be graphs. If $W \subseteq V$ and $F
\subseteq E$ then we call $H$ a \emph{subgraph} of $G$ and write $H
\leq G$. In other words, $H$ can be obtained from $G$ be removing
vertices and edges.

We say that $H$ is a \emph{minor} of $G$, written $H \preceq G$, if
there are disjoint connected nonempty subgraphs $(B_w)_{w \in W}$ in
$G$ such that for every edge $xy \in F$ there is an edge $ab \in E$
for some $a \in B_x$ and $b \in B_y$. The sets $(B_w)_{w \in W}$ are
called \emph{branch sets} of the minor $H$. Equivalently, $H \preceq
G$ if $H$ can be obtained by repeatedly contracting edges in a
subgraph of $G$.

A graph $H'$ is a \emph{subdivision} of a graph $H$ if it can be
obtained from $H$ be replacing edges with paths. If $H' \leq G$ for
some subdivision $H'$ of $H$ we say that $H$ is a \emph{topological
  subgraph} of $G$ and write $H \topmin G$. In this case there is an
injective mapping $\iota : W \to V$ and independent paths
$P_{\iota(u)\iota(v)}$ connecting $\iota(u)$ to $\iota(v)$ in $G$ for
$uv \in F$. The vertices in the image of $\iota$ are called
\emph{branch vertices}. Obviously $H \topmin G$ implies $H \preceq G$,
but the converse is not in general true.

\subsection{Logics}
\label{sec:logicprelim}

We will be dealing with finite structures over finite, relational
vocabularies. Thus a \emph{vocabulary} $\sigma$ is a finite set of
relation symbols $R$, each with an associated \emph{arity} $a(R)$, and
a $\sigma$-\emph{structure} $A$ consists of a finite set $V(A)$ (the
\emph{universe}) and relations $R(A) \subseteq A^{a(R)}$ for all $R
\in \sigma$. For vocabularies $\sigma \subseteq \tau$ and a
$\sigma$-structure $A$, a $\tau$-\emph{expansion} $B$ is a
$\tau$-structure with $V(A) = V(B)$ and $R(B) = R(A)$ for all $R \in
\sigma$.

The \emph{Gaifman graph} of a structure $A$ is the graph with vertex
set $V(A)$ and edge set
\[
\{ xy \suchthat x\text{ and }y\text{ appear together in some relation }R(A) \}.
\]
When applying graph-theoretic notions such as planarity to relational
structures, we mean that the corresponding Gaifman graph has the said
property.

We use standard definitions for first-order logic (FO),
cf.~\cite{eft96,ebbflu99,lib04}. In particular, $\bot$ and $\top$
denote false and true, respectively. Let $\sigma$ be a vocabulary and
$\suc \not\in \sigma$ a new binary relation symbol. We set
$\sigma_\suc := \sigma\cup\{\suc\}$ and say that $\suc$ is interpreted
by a \emph{successor relation} in a $\sigma_\suc$-structure $B$ if
$\suc(B)$ is the graph of a cyclic permutation on $V(B)$. An
$\FO[\sigma_\suc]$-formula $\varphi$ is called
\emph{successor-invariant} if for all $\sigma$-structures $A$ and all
$\sigma_\suc$-expansions $B,B'$ of $A$ in which $\suc$ is interpreted
by a successor relation we have
\[
B \models \varphi \quad\Leftrightarrow\quad
B' \models \varphi,
\]
when all free variables of $\varphi$ are interpreted identically in
$B$ and $B'$. In this case we say that $A \models \varphi$ if $B
\models \varphi$ for one such expansion $B$ (equivalently for all such
expansions).

Note that another common definition of successor relation is to
require $\suc(A)$ to be of the form
\[
\{ (a_{1},a_{2}), (a_{2},a_{3}),\ldots,(a_{n-1},a_{n}) \}
\]
for some enumeration $V(A) := \{ a_1,\ldots,a_n \}$ of the elements of
$V(A)$. This differs from our definition in that we require $(a_n,a_1)
\in \suc(A)$ as well, eliminating the somewhat artificial status of
the first and last element. The expressive power of
successor-invariant FO is not affected by this, though the quantifier
rank of formulas may change.

Order-invariant first-order logic is defined analogously to
successor-invariant $\FO$, by allowing the use of a binary relation
$\leq$ which is interpreted as a linear order and demanding the truth
value of a formula to be independent of the chosen linear order.

\section{Model-Checking for Successor-Invariant First-Order Logic}

The main result of this paper is the following:
\begin{theorem}
  \label{thm:mainthm}
  There is an algorithm $\mathds{A}$ with the following properties:
  Let $H$ be a finite graph, $\sigma$ a relational vocabulary, $\varphi
  \in \FO[\sigma_{\mathrm{succ}}]$ a successor-invariant formula, and
  $A$ a $\sigma$-structure whose Gaifman graph does not contain $H$ as
  a topological subgraph. Then on input $H$, $A$ and $\varphi$ the
  algorithm $\mathds{A}$ checks whether
  \[
  A \models \varphi
  \]
  in time $f(\absval{V(H)} + \absval{\varphi})\cdot \absval{V(A)}^c$ for some
  computable function $f$ and $c \in \N$.
\end{theorem}

Note that model-checking for first-order logic on nowhere dense
classes of graphs is possible in time $f(\absval{\varphi})\cdot
\absval{V(A)}^{1+\epsilon}$ for arbitrarily small $\epsilon > 0$ by a
result of Grohe et al.~\cite{gks14}. Even though a representation of a
structure $A$ in computer memory is likely to induce a linear order on
the elements of $V(A)$, making this linear order or its successor
relation accessible to the formula $\varphi$ potentially complicates
the model-checking problem. In particular, adding the cycle
corresponding to this linear order (or any other cycle through the
whole graph) to $A$ may introduce new shallow minors.

The proof of Theorem~\ref{thm:mainthm} is based on the following two lemmas:
\begin{lemma}
  \label{lem:kwalktopsub}
  For every finite graph $H$ there are constants $k \in \N$ and $c \in
  \N$ such that for every graph $G$ which does not contain $H$ as a
  topological subgraph there is a graph $G'$ and a $k$-walk $w :
  [\ell] \to V(G')$ through $G'$ such that $G'$ is obtained from $G$
  by only adding edges and $G'$ does not contain $K_c$ as a
  topological subgraph. Furthermore, $k$, $c$, $G'$ and $w$ can be
  computed, given $G$ and $H$, in time $f(\absval{V(H)})\cdot
  \absval{V(G)}^d$ for some computable function $f$ and $d \in \N$.
\end{lemma}
\begin{lemma}
  \label{lem:succfromkwalk}
  Let $\sigma$ be a finite relational vocabulary, $A$ a finite
  $\sigma$-structure, and $w:[\ell] \to V(A)$ a $k$-walk through the
  Gaifman graph of $A$.

  Then there is a finite relational vocabulary $\sigma_k$ and a
  first-order fomula $\varphi^{(k)}_{\suc}(x,y)$, both depending only
  on $k$, and a $(\sigma\cup\sigma_k)$-expansion $A'$ of $A$ which can
  be computed from $A$ and $w$ in polynomial time, such that
  \begin{itemize}
  \item The Gaifman-graphs of $A'$ and $A$ are the same,
  \item $\varphi^{(k)}_{\suc}$ defines a successor relation on $A'$.
  \end{itemize}
\end{lemma}

Lemma~\ref{lem:succfromkwalk} is taken from~\cite[Lemma~4.4]{ekk13}
and has been proved there. We will prove Lemma~\ref{lem:kwalktopsub}
in Section~\ref{sec:kwalktopsub}. The proof of
Theorem~\ref{thm:mainthm} then is a combination of the above lemmas:
\begin{proof}[Proof of Theorem~\ref{thm:mainthm}]
  Given a $\sigma$-structure $A$, a successor-invariant
  $\sigma_{\mathrm{succ}}$-formula $\varphi$ and a graph $H$ which is
  not a topological subgraph of the Gaifman graph of $A$, we first
  compute the Gaifman graph $G$ of $A$. Using the algorithm of
  Lemma~\ref{lem:kwalktopsub} we then compute a $k$-walk $w : [\ell]
  \to V(A)$ through a supergraph $G'$ of $G$ which excludes some
  clique $K_{c'}$ as a topological subgraph.

  Let $E$ be a binary relation symbol. We expand $A$ to a
  $(\sigma\cup\{E\})$-structure $A'$ by setting
  \[
  E(A') := \{ (w(i),w(i+1)) \suchthat i \in [\ell-1] \} \cup \{
  (w(\ell),w(1)) \}.
  \]
  Then $G'$ is the Gaifman graph of $A'$, which by
  Lemma~\ref{lem:kwalktopsub} excludes $K_c$ as a topological
  subgraph.

  Using Lemma~\ref{lem:succfromkwalk} we compute, for a suitable $\tau
  \supseteq \sigma$, a $\tau$-expansion $A''$ of $A'$ and an
  $\FO[\tau]$-formula $\varphi^{(k)}_{\mathrm{succ}}(x,y)$ which
  defines a successor relation on $A''$. We substitute
  $\varphi^{(k)}_{\mathrm{succ}}(x,y)$ for $\mathrm{succ}xy$ for all
  occurences of $\mathrm{succ}$ in $\varphi$, obtaining an
  $\FO[\tau]$-formula $\tilde\varphi$ such that
  \[
  A'' \models \tilde\varphi \quad\Leftrightarrow\quad
  (A,S) \models \varphi
  \]
  where $S$ the successor relation defined by
  $\varphi^{(k)}_{\mathrm{succ}}$. Note
  $\varphi^{(k)}_{\mathrm{succ}}$ and $\tau$ depend only on $H$.

  Since the Gaifman graph $G''$ of $A''$ excludes $H$ as a topological
  subgraph, there is a class $\CCC$ of graphs of bounded expansion
  such that $G'' \in \CCC$. We can therefore use {Dvo\v r\'ak} et
  al.'s model-checking algorithm for $\FO$ on $\CCC$ to check whether
  \[
  A'' \models \tilde\varphi
  \]
  in time linear in $\absval{A}$.
\end{proof}

\section{$k$-walks in Graphs with Excluded Topological Subgraphs}
\label{sec:kwalktopsub}

In this section we will prove Lemma~\ref{lem:kwalktopsub}. Given a
graph $G$ which excludes a graph $H$ as a topological subgraph, as a
first step towards constructing a supergraph $G'$ with a $k$-walk we
compute a tree-decomposition of $G$ into graphs which exclude $H$ as a
minor and graphs of almost bounded degree:
\begin{theorem}[Theorem~4.1 in~\cite{gromar12+}]
  \label{thm:topminstructure}
  For every $k \in \N$ there exists a constant $c = c(k) \in \N$ such
  that the following holds: If $H$ is a a graph on $k$ vertices and
  $G$ a graph which does not contain $H$ as a topological subgraph, then
  there is a tree-decomposition $(\TTT,\VVV)$ of $G$ of adhesion at
  most $c$ such that for all $t \in \TTT$
  \begin{itemize}
  \item  ${\bar\VVV}_t$ has at most $c$ vertices of degree larger than
    $c$, \emph{or}
  \item  ${\bar\VVV}_t$ excludes $K_c$ as a minor.
  \end{itemize}
  Furthermore, there is an algorithm that, given graphs $G$ of size
  $n$ and $H$ of size $k$ computes such a decomposition in time
  $f(k)\cdot n^{O(1)}$ for some computable function $f:\N\to\N$.
\end{theorem}

For the rest of this section we assume a graph $G = (V,E)$ together
with a tree-decomposition $(\TTT,\VVV)$ satisfying the properties of
Theorem~\ref{thm:topminstructure} as given. We will construct
$k$-walks through each of the bags of this decomposition, for a
suitable $k$ depending only on $H$, suitably adding edges within the
bags in a way that will not create large topological subgraphs. We
will then connect these $k$-walks to obtain a $k'$-walk through all of
$G$, carefully adding further edges where necessary.

If $s,t \in \TTT$ are neighbours in $\TTT$ we will connect the
$k$-walk through $\VVV_s$ and the $k$-walk through $\VVV_t$ by
joining them along a suitably chosen vertex $v \in \VVV_s \cap
\VVV_t$. Since the resulting walk may visit $v$ a total of $k+1$
times, we must be careful not to select the same vertex $v$ more than
a bounded number of times.

We first pick an arbitrary tree node $r \in \TTT$ as the root of the
tree-decomposition. Notions such as parent and sibling nodes are meant
with respect to this root node $r$. For a node $t \in \TTT$ we define
its adhesion $\alpha_t \subseteq \VVV_t$ as
\[
\alpha_t := \begin{cases}
  \emptyset & \text{if }t = r
  \\
  \VVV_s \cap \VVV_t & \text{if }s\text{ is the parent of }t.
  \end{cases}
\]
By adding the necessary edges within the bags we may assume that each
$\VVV_t$ is identical to its torso, in other words we may assume that
$G[\alpha_t]$ is a clique for each $t \in \TTT$.

\subsection{Computing the $k$-walks $w_t$}

Let $s,t \in \TTT$ be nodes such that $s$ is the parent of $t$. It may
happen that $\alpha_s \cap \alpha_t \not = \emptyset$, and that in fact
some vertex $v \in V$ appears in an unbounded number of bags. Since we
are only allowed to visit each vertex a bounded number of times, we
first compute, for $t \in \TTT$, a $k$-walk $w_t$ through a suitable
supergraph of $\VVV_t \setminus \alpha_t$.

If ${\bar\VVV}_t$ contains only $c$ vertices of degree larger than $c$
we choose an arbitrary enumeration $v_1,\ldots,v_{\ell}$ of
$\VVV_t\setminus \alpha_t$ and add edges
\[
v_1v_2, v_2v_3, \ldots,v_{\ell-1}v_\ell,v_{\ell}v_1
\]
to $G$. This will increase the degree of each vertex by at most $2$,
so there are still at most $c$ vertices of degree larger than
$c+2$. We set
\[\begin{split}
w_t : [\ell] &\to \VVV_t
\\
i &\mapsto v_i
\end{split}
\]
for these bags.

If, on the other hand, ${\bar\VVV}_t$ exludes a clique $K_c$ as
a minor, we invoke the following lemma on the graph $\VVV_t \setminus
\alpha_t$:
\begin{lemma}[Lemma~3.3 in \cite{ekk13}]
  \label{lem:kwalkexcl}
  For every natural number $c$ there are $k,c'\in\N$ such that: If $G
  = (V,E)$ is a graph which does not contain a $K_c$-minor, then there
  is a supergraph $G' = (V,E')$ obtained from $G$ by possibly adding
  edges such that $G'$ does not contain a $K_{c'}$-minor and there is a
  $k$-walk $w$ through $G'$. Moreover, $G'$ and $w$ can be found in
  polynomial time for fixed $c$.
\end{lemma}
Since we ignore the vertices in $\alpha_t$ when computing the
$k$-walk $w_t$, it may happen that the resulting supergraph of
${\bar\VVV}_t$ \emph{does} contain a $K_{c'}$-minor. However, the largest
possible clique minor is still of bounded size, because
$\absval{\alpha_t} \leq c$:
\begin{lemma}
  Let $G = (V,E)$ be a graph such that $K_{c'} \not\preceq G$, and let $G
  \oplus K_c$ be the graph with vertex set $V' = V \cup [c]$ and edge
  set
  \[
  E' = E \cup \binom{[c]}{2} \cup \{ va \suchthat v \in V, a \in [c] \}.
  \]
  In other words, $G \oplus K_c$ is the disjoint sum of $G$ and $K_c$
  plus edges between all vertices of $G$ and all vertices of
  $K_c$. Then $K_{c+c'} \not\preceq G \oplus K_c$.
\end{lemma}
\begin{proof}
  Otherwise let $X_1,\ldots,X_{c+c'}$ be the branch sets of a
  $K_{c+c'}$-minor in $G \oplus K_c$. At most $c$ of the sets contain
  vertices of the added $K_c$-clique. The remaining sets form the
  branch sets of a $K_c$-minor in $G$, contradicting the assumption
  that $K_c \not\preceq G$.
\end{proof}

\subsection{Connecting the $k$-walks}

We still need to connect the $k$-walks through the individual bags of
$(\TTT,\VVV)$ to obtain a single $k'$-walk through the whole graph,
for some $k'$ to be determined below. This is the most complicated
part of our construction, since we must guarantee that no vertex is
visited more than $k'$ times by the resulting walk, and that no large
topological subgraphs are created.

In the case of graphs excluding some fixed minor, the Graph Structure
Theorem guarantees the existence of a tree-decomposition into nearly
embeddable graphs such that neighbouring bags intersect only in apices
and vertices lying on some face or vortex of their near embeddings,
and this was used in~\cite{ekk13} to select vertices from the adhesion
sets of bags in a suitable way. Since the decomposition theorem for
graphs excluding a topological minor does not provide this kind of
information, we need a different approach here. Our method for
selecting vertices along which to connect the $k$-walks relies on the
fact that sparse graphs are \emph{degenerate}, i.e. every subgraph of
a sparse graph contains some vertex of small degree.

In connecting the walks $w_t$, we will proceed down the tree
$\TTT$. At any point in the process we keep a set $D \subseteq \TTT$
and a walk $w$ such that
\begin{itemize}
\item $D$ is a connected subset of $\TTT$,
\item the $k'$-walk has been constructed in $\bigcup_{t \in D} \VVV_t$,
\item if $s \in D$ and $s'$ is a sibling of $D$ then also $s' \in D$,
\item $w$ is a $k'$-walk through $\bigcup_{t \in D} \VVV_t$, and if
  $t \in D$ has a child $s \not\in D$, then the vertices in $\VVV_t
  \setminus \alpha_t$ are visited at most $k+1$ times by $w$.
\end{itemize}
We start with $D = \{r\}$ and $w = w_r$, where $r$ is the root of
$\TTT$. This is easily seen to satisfy all of the above conditions.

Now let $t \in D$ be a node whose children $s_1,\ldots,s_{n}$ are not
in $D$. We let
\[
C_i := \alpha_{s_i} \setminus \alpha_t
\]
be the adhesion set of $s_i$ with all vertices of the adhesion set of
$t$ removed. If $C_i = \emptyset$ then $s_i$ can be made a sibling of
$t$ (rather than a child), so we assume that all $C_i$ are
nonempty. Since the properties of $(\TTT,\VVV)$ are guaranteed for
the torsos of the bags we may assume that $G[C_i]$ is a clique for
each $i$ and that $w$ visits the vertices of $\bigcup C_i$ at most $k+1$
times.

It may happen that $C_i = C_j$ for some $i \not= j$. To deal with
this, assume that
\[
C_1 = C_2 = \cdots = C_m \not= C_i\text{ for i > m}.
\]
For each $i = 1,\ldots,m$ we choose an edge $u_iv_i \in E(\VVV_{s_i})$
which is traversed by the walk $w_{s_i}$ in the direction from $u_i$
to $v_i$ at some point. We add edges
\[
v_iu_{i+1}\text{ for }i=1,\ldots,m-1\text{ and }v_mu_1
\]
and connect the walks $w_{s_1},\ldots,w_{s_m}$ along these
edges. Because $w_{s_i}$ is a walk through $\VVV_{s_i} \setminus
\alpha_{s_i}$, we have
\[
u_i,v_i \in \VVV_{s_j} \quad\Leftrightarrow\quad i = j
\]
for $i,j = 1,\ldots,m$. To accomodate for the extra edges, we add the
vertices $u_i$ and $v_i$ to $\VVV_t$, and therefore to $\alpha_{s_i}$
and $C_i$. Since these vertices together with the added edges form an
isolated cycle
\[
u_1v_1u_2v_2\ldots u_mv_mu_1
\]
in $\VVV_t$, no new topological subgraphs are created by this. The
maximal adhesion of $(\TTT,\VVV)$ is still bounded by $c+2$.

Therefore we now assume that the cliques $C_1,\ldots,C_{n}$ are all
distinct. It remains to find a function
\[
f : [n] \to V
\]
such that
\begin{itemize}
\item $f(i) \in C_i$ for all $i$, and
\item $\absval{f^{-1}(v)} \leq M$ for all $v \in V$ and some constant
  $M$ depending only on $H$.
\end{itemize}

We define the function $f$ iteratively on larger subsets of $[n]$ as
follows: Let $\tilde G$ be the subgraph of $G$ induced on the union of
all $C_i$:
\[
\tilde G = G\left[\bigcup_i C_i\right].
\]
We show that $\tilde G$ contains a vertex of degree (in $\tilde G$) at
most $d$, for some constant $d$ depending only on the constant $c$
from Theorem~\ref{thm:topminstructure} (and therefore only on the
excluded topological subgraph $H$ we started with). If $\VVV_t$
contains only $c$ vertices of degree larger than $c$ then this is true
with $d = c$. If $\VVV_t$ excludes some clique $K_c$ as a minor we use
the fact that these graphs are $d$-degenerate for some $d$ depending
only on $c$. In fact, by Theorem~7.2.1 in \cite{diestel} there is a
constant $d$ such that if the average degree of $\tilde G$ is at least
$d$, then $K_c \topmin \tilde G$ and therefore $K_c \preceq \tilde G$.

In both cases there is a $v \in \bigcup_i C_i$ which has degree at
most $d$ in $\tilde G$. But any such $v$ can only be in at most
\[
M := \binom{d}{0} + \binom{d}{1} + \cdots + \binom{d}{c+1}
\]
of the cliques $C_i$, which all satisfy $\absval{C_i} \leq c+2$. It is
therefore safe to define
\[
f(i) := v\text{ for all }i\in [n]\text{ such that }v \in C_i.
\]
We remove these cliques and iterate until no cliques remain.

\begin{figure}
  \begin{center}
    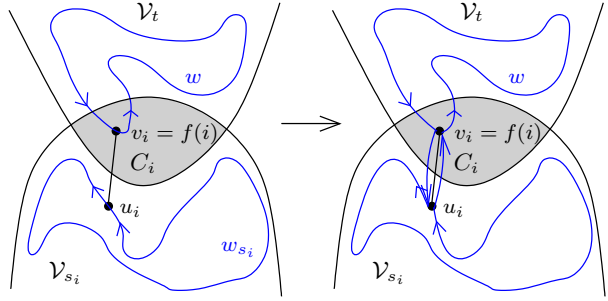
  \end{center}
  \caption{Connecting the individual $k$-walks}  
  \label{fig:connect}
\end{figure}

Once the function $f$ has been found we connect the walk $w_t$ through
$\VVV_t$ with the walks $w_{s_i}$ through the bags $V_{s_i}$. Let $w :
[\ell] \to V$ be the walk constructed so far. For each $i \in [n]$ let
$v_i = f(i) \in C_i$ be the vertex chosen by $f$, and let $u_i \in
\VVV_{s_i} \setminus \alpha_{s_i}$ be a neighbour of $v_i$. If no such
neighbour exists it is safe to create one by adding an edge between
$v_i$ and an arbitrary vertex of $\VVV_{s_i} \setminus
\alpha_{s_i}$. We now extend the walk $w$ by inserting the $k$-walk
$w_{s_i}$ along the edge $v_iu_i$ when $v_i$ is first visited by
$w$. This increases the number of times $v_i$ and $u_i$ are visited by
one each (cf. Figure~\ref{fig:connect}).

After inserting all walks $w_{s_1},\ldots,w_{s_n}$ we set
\[
D := D \cup \{ s_1,\ldots,s_n \}
\]
and repeat the process until $D = \TTT$. Note that the resulting walk
is a $(k + M + 1)$-walk through the supergraph $G'$ of $G$ obtained by
adding edges to $G$.

\subsection{Topological Subgraphs in $G'$}

By now we have a supergraph $G'$ of $G$, obtained by only adding
edges, and a $k' = (k+M+1)$-walk $w : [\ell] \to V(G')$ through this
supergraph. Furthermore, there is a $c' = c'(H)$ depending only on
(the size of) $H$ and a tree-decomposition $(\TTT,\VVV)$ of $G'$ such
that if $s,t \in \TTT$ then $\absval{\VVV_s \cap \VVV_t} c'$ and for
all $t \in \TTT$
\begin{itemize}
\item ${\bar\VVV}_t$ has at most $c'$ vertices of degree larger
  than $c'$ or
\item ${\bar\VVV}_t$ excludes $K_{c'}$ as a minor.
\end{itemize}

We show that this implies $K_{c'+2} \not \topmin G'$: Assume for a
contradiction that $K_{c'+2} \topmin G$, and let $v_1,\ldots,v_{c'+2}$
be the branch vertices of a $K_{c'+2}$-subdivision in $G$. Then there
is a $t \in \TTT$ such that $\{ v_1,\ldots,v_{c'+2} \} \subseteq
\VVV_t$: Otherwise choose $i < j$ and $t \not = t'$ so that
\[
v_i \in \VVV_t \setminus \VVV_{t'}
\quad\text{and}\quad
v_j \in \VVV_{t'} \setminus \VVV_t.
\]
Then, since the adhesion of $(\TTT,\VVV)$ is at most $c'$, there is a
set $S \subseteq V$ of size at most $c'$ separating two branch
vertices, which is not possible in a $(c'+2)$-clique.

Now let $t \in \TTT$ be a tree node for which $\VVV_t$ contains all
branch vertices. For $i < j$, let $P_{ij}$ be the path in $G$
connecting $v_i$ and $v_j$. If all vertices on this path are in
$\VVV_t$ we are done. Otherwise we may shorten this path to get a path
$P_{ij}'$ connecting $v_i$ and $v_j$ in the torso of $\VVV_t$. Thus
\[
K_{c'+2} \topmin \VVV_t.
\]

But none of the bags $\VVV_t$ can contain $K_{c'+2}$ as a topological
subgraph: Since $K_{c'+2} \topmin \VVV_t$ implies $K_{c'+2} \preceq
\VVV_t$ which in turn implies $K_{c'} \preceq \VVV_t$, none of the
bags excluding $K_{c'}$ as a minor can contain $K_{c'+2}$ as a
topological subgraph. But if $K_{c'+2} \topmin \VVV_t$ then there must
be at least $c'+2$ vertices of degree at least $c'+1$, namely the
branch vertices of the image of a subdivision of $K_{c'+2}$. We
conclude that $K_{c'+2} \not\topmin G'$.

\section{Dense Graphs}
\label{sec:dense}

While model-checking for first-order logic has been studied rather
thoroughly for sparse graph classes, few results are known for dense
graphs:
\begin{itemize}
\item On classes of graphs with bounded clique-width (or,
  equivalently, bounded rank-width; cf.~\cite{OumSey06}),
  model-checking even for monadic second-order logic has been shown to
  be fpt by Courcelle et al.~\cite{CourcelleMakRot00}.
\item More recently, model-checking on coloured posets of bounded
  width has been shown to be in fpt for existential $\FO$ by Bova et
  al.~\cite{bova2015model} and for all of $\FO$ by Gajarský et
  al.~\cite{gajarsky2015fo}.
\end{itemize}

Both of these results extend to order-invariant $\FO$, and therefore
also to successor-invariant $\FO$. For bounded clique-width, this has
already been shown by Engelmann et
al. in~\cite[Thm.~4.2]{EngelmannKS12}. For posets of bounded width we
give a proof here. We first review the necessary definitions:
\begin{definition}
  A \emph{partially ordered set (poset)} $(P,\leq^P)$ is a set $P$
  with a reflexive, transitive and antisymmetric binary relation
  $\leq^P$. A \emph{chain} $C \subseteq P$ is a totally ordered
  subset, i.e. for all $x, y \in C$ one of $x \leq^P y$ and $y \leq^P
  x$ holds. An \emph{antichain} is a set $A \subseteq P$ such that if
  $x \leq^P y$ for $x,y \in A$ then $x = y$. The \emph{width} of
  $(P,\leq^P)$ is the maximal size $|A|$ of an antichain $A \subseteq
  P$. A \emph{coloured} poset is a poset $(P,\leq^P)$ together with a
  function $\lambda : P \to \Lambda$ mapping $P$ to some set $\Lambda$
  of \emph{colours}. By $\Vert P\Vert$ we denote the length of a
  suitable encoding of $(P,\leq^P)$.
\end{definition}
We will need Dilworth's Theorem, which relates the width of a poset to
the minimum number of chains needed to cover the poset:
\begin{theorem}[Dilworth's Theorem]
  Let $(P,\leq^P)$ be a poset. Then the width of $(P,\leq^P)$ is equal
  to the minimum number $k$ of disjoint chains $C_i,\ldots,C_k
  \subseteq P$ needed to cover $P$, i.e. such that $\bigcup_i C_i =
  P$.
\end{theorem}
A proof can be found in~\cite[Sec.~2.5]{diestel}. Moreover, by a
result of Felsner et al.~\cite{felsner2003recognition}, both the width
$w$ and a set of chains $C_1,\ldots,C_w$ covering $P$ can be computed
from $(P,\leq_P)$ in time $O(w\cdot \Vert P \Vert)$.

With this, we are ready to prove the following:
\begin{theorem}
  There is an algorithm which, on input a coloured poset $(P,\leq^P)$
  with colouring $\lambda : P \to \Lambda$ and an order-invariant
  first-order formula $\varphi$, checks whether $P \models \varphi$ in
  time $f(w,|\varphi|) \cdot \Vert P\Vert^2$ where $w$ is the width of
  $(P,\leq^P)$.
\end{theorem}
\begin{proof}
  Using the algorithm of~\cite{felsner2003recognition}, we compute a
  chain cover $C_1,\ldots,C_w$ of $(P,\leq^P)$. To obtain a linear
  order on $P$, we just need to arrange the chains in a suitable
  order, which can be done by colouring the vertices with colours
  $\Lambda \times [w]$ via
  \[
  \lambda'(v) = (\lambda(v), j)\text{ for }v \in C_j.
  \]
  Then
  \begin{multline*}
  \varphi_\leq(x,y) := \bigg(\bigvee_{\substack{\lambda_x,\lambda_y \in \Lambda,\\i <
    j}} (\lambda'(x) = (\lambda_x,i) \wedge \lambda'(y) = (\lambda_y,j)
  \bigg)
  \vee
  \\
  \bigg(\bigvee_{\substack{\lambda_x,\lambda_y \in \Lambda,\\i \in [w]}} \lambda'(x)
  = (\lambda_x,i)
  \wedge
  \lambda'(y) = (\lambda_y,i)
  \wedge
  x \leq y
  \bigg)
  \end{multline*}
  defines a linear order on $(P,\leq_P)$ with colouring
  $\lambda'$. After substituting $\varphi_\leq$ for $\leq$ in
  $\varphi$ we may apply Gajarský et al.'s algorithm to check whether
  $P \models \varphi$.
\end{proof}

\section*{Conclusion and Further Research}

We have shown that model-checking for successor-invariant first-order
logic is fixed-parameter tractable on classes of graphs excluding some
fixed graph $H$ as a topological subgraph. This extends previous
results showing tractibility on planar graphs~\cite{EngelmannKS12} and
graphs with excluded minors~\cite{ekk13}. For dense graphs, we showed
how the recent model-checking algorithm by Gajarský et
al.~\cite{gajarsky2015fo} can be adapted to order-invariant $\FO$.

This prompts for further generalisation in two ways: First, can we
close the gap between plain first-order logic and its
successor-invariant counterpart? Next steps could be graph classes
with bounded expansion or with locally excluded minors. However, no
structure theorem comparable to those of Robertson and Seymour and of
Grohe and Marx are known for these graph classes.

Another interesting open question is whether model-checking for
order-invariant first-order logic is tractable on any of the classes
depicted in Figure~\ref{fig:sparseclasses}. Since the Gaifman graph of
a linearly ordered structure is a clique, there is no hope of finding
a ``good'' linear order which can be added to the input structure
without destroying the desirable properties of its Gaifman graph. As
shown in~\cite{GroheS00}, order-invariant first-order logic has a
Gaifman-style locality property (see also~\cite{AndersonMSS12}). It
is, however, not at all clear how this could be turned into an
efficient model-checking algorithm. In particular, no variant of
Gaifman normal form is known for this logic.

\end{document}

%% file: connect.pdf_t
\begin{picture}(0,0)%
\includegraphics{connect.pdf}%
\end{picture}%
\setlength{\unitlength}{4144sp}%
\begingroup\makeatletter\ifx\SetFigFont\undefined%
\gdef\SetFigFont#1#2#3#4#5{%
  \reset@font\fontsize{#1}{#2pt}%
  \fontfamily{#3}\fontseries{#4}\fontshape{#5}%
  \selectfont}%
\fi\endgroup%
\begin{picture}(3589,1822)(519,827)
\put(1294,2502){\makebox(0,0)[lb]{\smash{{\SetFigFont{9}{10.8}{\familydefault}{\mddefault}{\updefault}{\color[rgb]{0,0,0}$\VVV_t$}%
}}}}
\put(769,948){\makebox(0,0)[lb]{\smash{{\SetFigFont{9}{10.8}{\familydefault}{\mddefault}{\updefault}{\color[rgb]{0,0,0}$\VVV_{s_i}$}%
}}}}
\put(1584,2117){\makebox(0,0)[lb]{\smash{{\SetFigFont{8}{9.6}{\familydefault}{\mddefault}{\updefault}{\color[rgb]{0,0,1}$w$}%
}}}}
\put(1193,1337){\makebox(0,0)[lb]{\smash{{\SetFigFont{8}{9.6}{\familydefault}{\mddefault}{\updefault}{\color[rgb]{0,0,0}$u_i$}%
}}}}
\put(1256,1597){\makebox(0,0)[lb]{\smash{{\SetFigFont{9}{10.8}{\familydefault}{\mddefault}{\updefault}{\color[rgb]{0,0,0}$C_i$}%
}}}}
\put(1801,1109){\makebox(0,0)[lb]{\smash{{\SetFigFont{8}{9.6}{\familydefault}{\mddefault}{\updefault}{\color[rgb]{0,0,1}$w_{s_i}$}%
}}}}
\put(1261,1784){\makebox(0,0)[lb]{\smash{{\SetFigFont{8}{9.6}{\familydefault}{\mddefault}{\updefault}{\color[rgb]{0,0,0}$v_i = f(i)$}%
}}}}
\put(3229,2502){\makebox(0,0)[lb]{\smash{{\SetFigFont{9}{10.8}{\familydefault}{\mddefault}{\updefault}{\color[rgb]{0,0,0}$\VVV_t$}%
}}}}
\put(2704,948){\makebox(0,0)[lb]{\smash{{\SetFigFont{9}{10.8}{\familydefault}{\mddefault}{\updefault}{\color[rgb]{0,0,0}$\VVV_{s_i}$}%
}}}}
\put(3519,2117){\makebox(0,0)[lb]{\smash{{\SetFigFont{8}{9.6}{\familydefault}{\mddefault}{\updefault}{\color[rgb]{0,0,1}$w$}%
}}}}
\put(3128,1337){\makebox(0,0)[lb]{\smash{{\SetFigFont{8}{9.6}{\familydefault}{\mddefault}{\updefault}{\color[rgb]{0,0,0}$u_i$}%
}}}}
\put(3191,1597){\makebox(0,0)[lb]{\smash{{\SetFigFont{9}{10.8}{\familydefault}{\mddefault}{\updefault}{\color[rgb]{0,0,0}$C_i$}%
}}}}
\put(3196,1784){\makebox(0,0)[lb]{\smash{{\SetFigFont{8}{9.6}{\familydefault}{\mddefault}{\updefault}{\color[rgb]{0,0,0}$v_i = f(i)$}%
}}}}
\end{picture}%

%% file: arxiv.bbl
\begin{thebibliography}{10}

\bibitem{AndersonMSS12}
Matthew Anderson, Dieter van Melkebeek, Nicole Schweikardt, and Luc Segoufin.
\newblock Locality from circuit lower bounds.
\newblock {\em {SIAM} J. Comput.}, 41(6):1481--1523, 2012.

\bibitem{BenediktS09}
Michael Benedikt and Luc Segoufin.
\newblock Towards a characterization of order-invariant queries over tame
  graphs.
\newblock {\em J. Symb. Log.}, 74(1):168--186, 2009.

\bibitem{bova2015model}
Simone Bova, Robert Ganian, and Stefan Szeider.
\newblock Model checking existential logic on partially ordered sets.
\newblock {\em ACM Transactions on Computational Logic (TOCL)}, 17(2):10, 2015.

\bibitem{Courcelle90}
Bruno Courcelle.
\newblock Graph rewriting: An algebraic and logic approach.
\newblock In J.~{van Leeuwen}, editor, {\em Handbook of Theoretical Computer
  Science}, volume~2, pages 194 -- 242. Elsevier, 1990.

\bibitem{CourcelleMakRot00}
Bruno Courcelle, Johann Makowsky, and Udi Rotics.
\newblock Linear time solvable optimization problems on graphs of bounded
  clique-width.
\newblock {\em Theory of Computing Systems}, 33(2):125--150, 2000.

\bibitem{DawarGroKre07}
Anuj Dawar, Martin Grohe, and Stephan Kreutzer.
\newblock Locally excluding a minor.
\newblock In {\em Logic in Computer Science (LICS)}, pages 270--279, 2007.

\bibitem{diestel}
Reinhard Diestel.
\newblock {\em Graph Theory}.
\newblock Number 173 in GTM. Springer, 4th edition, 2012.

\bibitem{DowneyF98}
Rod Downey and Michael~R. Fellows.
\newblock {\em Parameterized Complexity}.
\newblock Springer, 1998.

\bibitem{DvorakKT13}
Zden{\v{e}}k Dvo{\v{r}}{\'a}k, Daniel Kr{\'a}l, and Robin Thomas.
\newblock Testing first-order properties for subclasses of sparse graphs.
\newblock {\em Journal of the ACM (JACM)}, 60(5):36, 2013.

\bibitem{ebbflu99}
Heinz-Dieter Ebbinghaus and Jörg Flum.
\newblock {\em Finite Model Theory}.
\newblock Perspectives in Mathematical Logic. Springer, 2nd edition, 1999.

\bibitem{eft96}
Heinz-Dieter Ebbinghaus, Jörg Flum, and Wolfgang Thomas.
\newblock {\em Mathematical Logic}.
\newblock Springer, 2nd edition, 1994.

\bibitem{EickmeyerEH14}
Kord Eickmeyer, Michael Elberfeld, and Frederik Harwath.
\newblock Expressivity and succinctness of order-invariant logics on
  depth-bounded structures.
\newblock In {\em Mathematical Foundations of Computer Science 2014 - 39th
  International Symposium, {MFCS} 2014, Budapest, Hungary, August 25-29, 2014.
  Proceedings, Part {I}}, pages 256--266, 2014.

\bibitem{ekk13}
Kord Eickmeyer, Ken-Ichi Kawarabayashi, and Stephan Kreutzer.
\newblock Model checking for successor-invariant first-order logic on
  minor-closed graph classes.
\newblock In {\em Proceedings of the 2013 28th Annual ACM/IEEE Symposium on
  Logic in Computer Science}, LICS '13, pages 134--142. IEEE Computer Society,
  2013.

\bibitem{EngelmannKS12}
Viktor Engelmann, Stephan Kreutzer, and Sebastian Siebertz.
\newblock First-order and monadic second-order model-checking on ordered
  structures.
\newblock In {\em Logics in Computer Science}, pages 275--284, 2012.

\bibitem{felsner2003recognition}
Stefan Felsner, Vijay Raghavan, and Jeremy Spinrad.
\newblock Recognition algorithms for orders of small width and graphs of small
  dilworth number.
\newblock {\em Order}, 20(4):351--364, 2003.

\bibitem{FlumG01}
J{\"o}rg Flum and Martin Grohe.
\newblock Fixed-parameter tractability, definability, and model-checking.
\newblock {\em SIAM J. Comput.}, 31(1):113--145, 2001.

\bibitem{FlumG06}
Jörg Flum and Martin Grohe.
\newblock {\em Parameterized Complexity Theory}.
\newblock Springer, 2006.
\newblock ISBN 3-54-029952-1.

\bibitem{FrickGro01}
Markus Frick and Martin Grohe.
\newblock Deciding first-order properties of locally tree-decomposable
  structures.
\newblock {\em Journal of the ACM}, 48:1148 -- 1206, 2001.

\bibitem{gajarsky2015fo}
Jakub Gajarský, Petr Hliněný, Daniel Lokshtanov, Jan Obdržálek, Sebastian
  Ordyniak, MS~Ramanujan, and Saket Saurabh.
\newblock Fo model checking on posets of bounded width.
\newblock In {\em Foundations of Computer Science (FOCS), 2015 IEEE 56th Annual
  Symposium on}, pages 963--974. IEEE, 2015.

\bibitem{Grohe07}
Martin Grohe.
\newblock Logic, graphs, and algorithms.
\newblock In E.Gr{\"a}del~T.Wilke J.Flum, editor, {\em Logic and Automata --
  History and Perspectives}. Amsterdam University Press, 2007.

\bibitem{gks14}
Martin Grohe, Stephan Kreutzer, and Sebastian Siebertz.
\newblock Deciding first-order properties of nowhere dense graphs.
\newblock In {\em Proceedings of the 46th Annual ACM Symposium on Theory of
  Computing}, STOC '14, pages 89--98. ACM, 2014.

\bibitem{gromar12+}
Martin Grohe and Dániel Marx.
\newblock Structure theorem and isomorphism test for graphs with excluded
  topological subgraphs.
\newblock {\em SIAM Journal on Computing}, 44(1):114--159, 2015.

\bibitem{GroheS00}
Martin Grohe and Thomas Schwentick.
\newblock Locality of order-invariant first-order formulas.
\newblock {\em ACM Trans. Comput. Logic}, 1(1):112--130, July 2000.

\bibitem{Kreutzer11}
Stephan Kreutzer.
\newblock Algorithmic meta-theorems.
\newblock In Javier Esparza, Christian Michaux, and Charles Steinhorn, editors,
  {\em Finite and Algorithmic Model Theory}, London Mathematical Society
  Lecture Note Series, chapter~5, pages 177--270. Cambridge University Press,
  2011.
\newblock a preliminary version is available at Electronic Colloquium on
  Computational Complexity (ECCC), TR09-147,
  http://www.eccc.uni-trier.de/report/2009/147.

\bibitem{lib04}
Leonid Libkin.
\newblock {\em Elements of Finite Model Theory}.
\newblock Texts in Theoretical Computer Science. Spinger-Verlag, 2004.

\bibitem{OumSey06}
Sang-Il Oum and Paul~D. Seymour.
\newblock Approximating clique-width and branch-width.
\newblock {\em Journal of Combinatorial Theory, Series B}, 96:514 -- 528, 2006.

\bibitem{Rossman07}
Benjamin Rossman.
\newblock Successor-invariant first-order logic on finite structures.
\newblock {\em J. Symb. Log.}, 72(2):601--618, 2007.

\bibitem{Seese96}
Detlef Seese.
\newblock Linear time computable problems and first-order descriptions.
\newblock {\em Mathematical {S}tructures in Computer Science}, 5:505--526,
  1996.

\bibitem{tutte}
William~T. Tutte.
\newblock {\em Graph Theory}, volume~21 of {\em Encyclopedia of Mathematics and
  Its Applications}.
\newblock Cambridge University Press, 2001.

\end{thebibliography}
